\documentclass[a4paper,UKenglish,autoref,thm-restate]{lipics-v2021}
\usepackage[nameinlink]{cleveref}

\usepackage{subcaption}
\usepackage{siunitx}
\usepackage{complexity}
\usepackage{xspace}
\usepackage{nicefrac}
\usepackage{mathtools}

\graphicspath{{./figures/}}

\nolinenumbers

\crefname{figure}{Figure}{Figures}
\crefname{theorem}{Theorem}{Theorems}
\crefname{lemma}{Lemma}{Lemmas}
\crefname{corollary}{Corollary}{Corollaries}
\crefname{section}{Section}{Sections}
\crefname{appendix}{Appendix}{Appendices}
\crefname{remark}{Remark}{Remarks}
\crefname{claim}{Claim}{Claims}
\crefname{conjecture}{Conjecture}{Conjectures}
\crefname{observation}{Observation}{Observations}

\newcommand{\BigO}{\mathcal{O}}

\newcommand{\good}{independent\xspace}
\newcommand{\suffixgood}{suffix-independent\xspace}
\newcommand{\out}{outer\xspace}

\newcommand{\paths}{\mathcal{P}}
\newcommand{\suffixgoodpaths}{\mathcal{P}_{\sin}}
\newcommand{\chordfreepaths}{\mathcal{P}^*}

\newcommand{\leveledges}{\ell}
\newcommand{\cuttingedges}{\ensuremath{\psi}}

\newcommand{\verticalLine}{\mathcal{L}}
\newcommand{\ccwchain}{m}

\newcommand{\conv}{\operatorname{conv}}
\newcommand{\layernumber}{L}
\newcommand{\layer}{L}
\newcommand{\propk}[1]{$#1$-layer property\xspace}
\newcommand{\prop}{\propk{k}}

\hideLIPIcs
\title{On the Connectivity of the Flip Graph of Plane Spanning Paths}

\author{Linda {Kleist}}{Department of Computer Science, TU Braunschweig, Germany}{kleist@ibr.cs.tu-bs.de}{https://orcid.org/0000-0002-3786-916X}{}

\author{Peter {Kramer}}{Department of Computer Science, TU Braunschweig, Germany}{kramer@ibr.cs.tu-bs.de}{https://orcid.org/0000-0001-9635-5890}{}

\author{Christian {Rieck}}{Department of Computer Science, TU Braunschweig, Germany}{rieck@ibr.cs.tu-bs.de}{https://orcid.org/0000-0003-0846-5163}{}

\authorrunning{Linda Kleist, Peter Kramer, and Christian Rieck}

\Copyright{Linda Kleist, Peter Kramer, and Christian Rieck}


\begin{document}

\maketitle
	
\begin{abstract}
Flip graphs of non-crossing configurations in the plane  are widely studied objects, e.g., flip graph of triangulations, spanning trees, Hamiltonian cycles, and perfect matchings.
Typically, it is an easy exercise to prove connectivity of a flip graph.
In stark contrast, the connectivity of the flip graph of plane spanning paths on point sets in general position {has been} an open problem for more than~16~years.

In order to provide new insights, we investigate certain induced subgraphs.
Firstly, we provide tight bounds on the diameter and the radius of the flip graph of spanning paths on points in convex position with one fixed endpoint.
Secondly, we show that so-called \emph{\suffixgood paths} induce a connected subgraph.
Consequently, to answer the open problem affirmatively, it suffices to show that each path can be flipped to some \suffixgood path.
Lastly, we investigate paths where one endpoint is fixed and provide tools to flip to \suffixgood paths.
We show that these tools are strong enough to show connectivity of the flip graph of plane spanning paths on point sets {with at most two convex~layers}.

\end{abstract}

\section{Introduction}
\label{sec:introduction}

The reconfiguration of objects by small local operations is
a basic concept that appears in various contexts, e.g., robot motion planning, morphing, network dynamics, discrete mathematics, and computer science~\cite{nishimuraIntroReconfiguration}.
A \emph{flip graph} captures the structure of the reconfiguration space for the problem at hand: 
The vertices of the graph correspond to the possible configurations, i.e., states of the system, while an edge indicates that two configurations can be transformed into each other by a small local modification --- a so-called~\emph{flip}.
\begin{figure}[b]
	\centering
	\includegraphics[page=2]{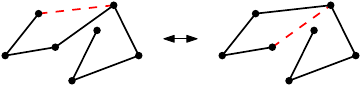}
	\caption{Plane spanning paths that can be transformed into each other by a flip.}
	\label{fig:flip}
\end{figure}
A famous example is the flip graph of triangulations of a convex $n$-point set where a flip operation replaces one diagonal by another one; this flip graph is the $1$-skeleton of the $(n-3)$-dimensional associahedron.
More generally, flip graphs of various non-crossing configurations on point sets are widely studied, e.g.,  triangulations, spanning~trees, Hamiltonian cycles, or perfect matchings; see~\cref{subsec:related-work}.
In this paper, we investigate the flip graph of \emph{plane spanning paths} where a flip consists of the insertion and deletion of one edge,
as illustrated in~\cref{fig:flip,fig:flipgraph-33}.
\begin{figure}[p]%
    \centering%
    \includegraphics[scale=.412,angle =90]{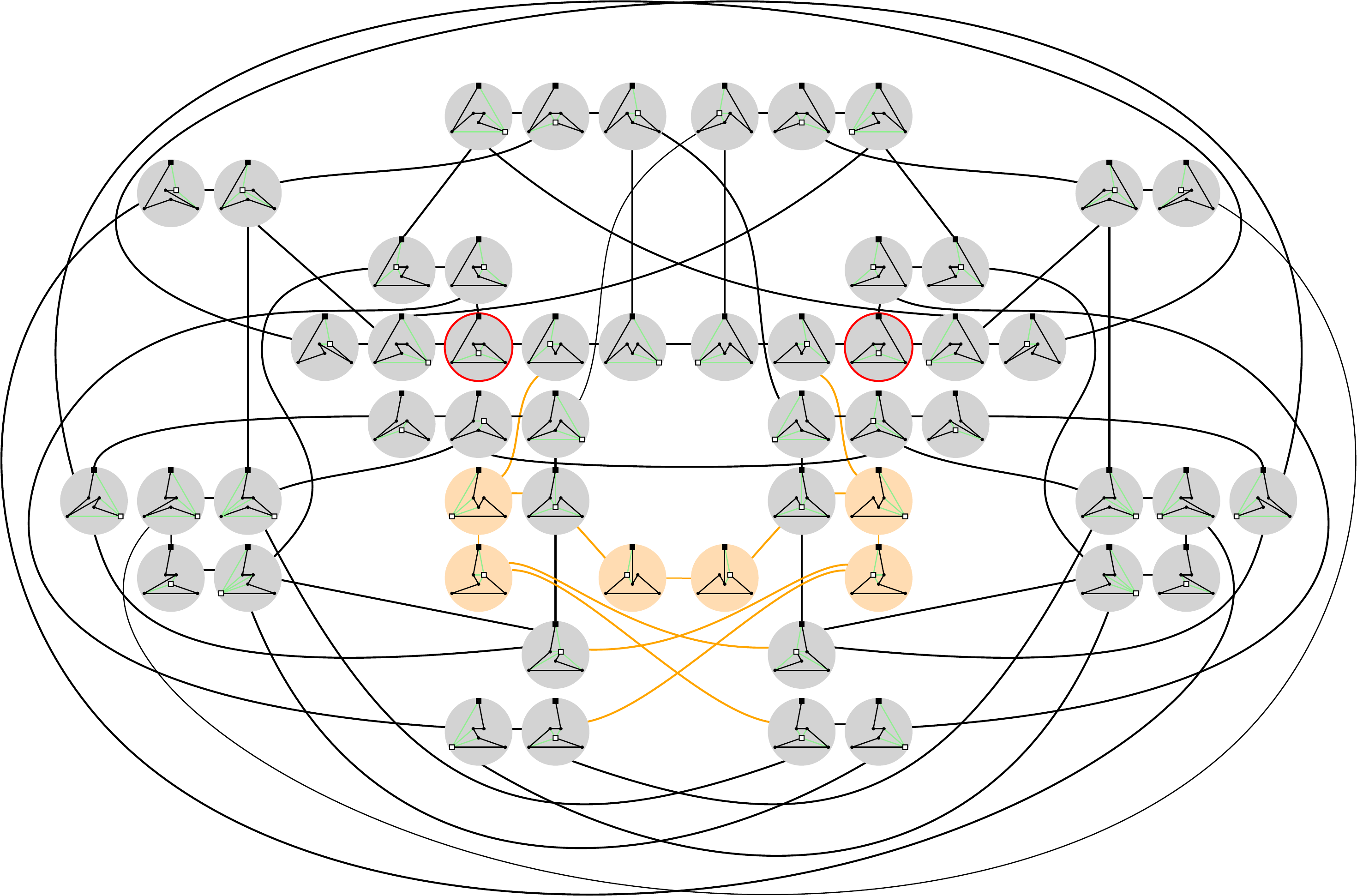}%
    \caption{
        The flip graph of plane spanning paths with a fixed start on a point set with two layers, where each layer contains three points.
        Six paths, highlighted in orange, are not \suffixgood.
        The two spirals, highlighted in red, have flip distance 5 and are a furthest pair.
    }
    \label{fig:flipgraph-33}%
\end{figure}%

Typically, one is interested in properties such as the diameter of the flip graph (the maximum length of a shortest reconfiguration sequence over all pairs), the distance between two given configurations, or whether the flip graph contains a Hamiltonian cycle, i.e., a cyclic listing of all configurations.
Usually, it is a simple exercise to prove that a flip graph is connected.
In stark contrast, the connectivity of the flip graph of plane spanning paths has been an open problem for more than 16 years~\cite{akl2007planar,bose2009flipsinplanar} and has been conjectured by Akl~et~al.~\cite{akl2007planar}.

\begin{conjecture}[Akl et al.~\cite{akl2007planar}]
    \label{conj:conn}
    For every point set $S$ in general position, the flip graph of plane spanning paths on $S$ is connected.
\end{conjecture}
The conjecture has been confirmed for some special cases, namely, point sets with up to~8 points~\cite{akl2007planar}, sets in convex position~\cite{akl2007planar}, and generalized double~circles~\cite{2023Aicholzer}.

In a spirit similar to \Cref{conj:conn}, Aichholzer et al.~\cite{2023Aicholzer} conjecture the connectivity of subgraphs induced by all paths with a fixed endpoint.
In fact, they show that~\Cref{conj:connFixed} implies~\Cref{conj:conn}.
\begin{conjecture}[Aichholzer et al.~\cite{2023Aicholzer}]\label{conj:connFixed}
    For every point set $S$ in general position and every point $s\in S$, the flip graph of plane spanning paths on $S$ starting at~$s$ is connected.
\end{conjecture}

\subsection{Our contribution}
\label{subsec:contributions}
In this work, we tackle~\Cref{conj:conn,conj:connFixed}.
For a point set $S$, we denote the set of all paths on $S$ by $\paths(S)$ and the subset of paths  starting at a fixed $s\in S$ by~$\paths(S,s)$.
Firstly, we settle \cref{conj:connFixed} for point sets in convex position.
More precisely, we provide tight bounds for the diameter and the radius.
Note that for~$n=1,2$, the flip graphs are trivial, as they consist of a single vertex.

\begin{restatable}{theorem}{convexBounds}
	\label{lem:convex-fixed-flipgraph}
	Let $S$ be a set of $n$ points in convex position and let $s\in S$.
    For~${n\geq3}$, the flip graph of~$\paths(S,s)$ has diameter $2n-5$ and radius $n-2$.
    Moreover, the spirals are exactly the centers of the flip graph.
\end{restatable}

Secondly, we show that the flip graph of so-called \emph{\suffixgood paths} is connected.
A~path is \emph{\suffixgood} if each suffix starts on the boundary of its convex hull; for precise definitions, we refer to~\cref{sec:defs}.
We denote the set of \suffixgood paths starting at~$s\in S$ by $\suffixgoodpaths(S,s)$.

\begin{restatable}{theorem}{thmSuffix}
    \label{thm:suffix-good} 
    For every point set $S$ in general position and every  point $s\in \layer_0$, the flip graph of $\suffixgoodpaths(S,s)$ is connected.
\end{restatable}
\cref{thm:suffix-good} may serve as a crucial building block to prove \cref{conj:conn,conj:connFixed} as it  suffices to flip each path to some \suffixgood path.
We investigate a point set $S$ based on its decomposition into $\layernumber(S)$ many \emph{convex layers}~$\layer_0,\layer_1, \dots$.
Using \cref{thm:suffix-good} alongside other tools, we conclude connectivity of the flip graph of plane spanning paths on each point set with at most two convex layers.
As a crucial tool, we consider paths that start at a fixed point on the convex~hull.

\begin{restatable}{theorem}{thmConvOut}
    \label{thm:2conv}
    Let $S$ be a point set in general position with $\layernumber(S)\leq 2$, and~$s\in \layer_0$.
    Then the flip graph of~$\paths(S,s)$ is connected.
\end{restatable}
We then use the insights from~\cref{thm:2conv} to settle~\Cref{conj:conn} for point sets with at most two convex layers.

\begin{restatable}{theorem}{thmFlipGraphConnected}
    \label{cor:2conv}
    Let $S$ be a point set in general position with $\layernumber(S)\leq 2$.
    Then the flip graph of~$\paths(S)$ is connected.
\end{restatable}

\subsection{Related work}
\label{subsec:related-work}
Flip graphs of non-crossing configurations on point sets in the plane are studied in various settings, e.g., for triangulations~\cite{Eppstein10,FelsnerKMS20,HouleHNR05,HurtadoNU99,Lawson72,LubiwP15,WagnerW22}, spanning paths~\cite{2023Aicholzer,akl2007planar,convexDiameter,hamilton}, spanning trees~\cite{AichholzerAH02,aichholzer2022reconfiguration,AvisF96,trees,bousquet2024reconfigurationSoCG,bousquet2023note,HernandoHMMN99,NicholsPTZ20}, Hamiltonian cycles~\cite{HernandoHH02}, and matchings~\cite{perfect-matchings,HouleHNR05,MilichMP21}.

In the following, we review the state of the art for the two most related settings, namely of plane spanning paths and plane spanning trees.

\subparagraph*{Plane spanning paths.}
While not much is known for the flip graph of plane spanning paths on general point sets, the case of convex point sets is pretty well understood.
Akl et al.~\cite{akl2007planar} show that the flip graph is connected and Chang and Wu~\cite{convexDiameter} provide tight bounds for the diameter, i.e., for a convex $n$-set, its flip graph has diameter $2n-5$ for $n=3,4$ and~$2n-6$ for all $n\geq5$.
Rivera-Campo and Urrutia-Galicia~\cite{hamilton} show Hamiltonicity of the flip graph, and its chromatic number is shown to be $n$ by Fabila-Monroy et al.~\cite{chromatic}.

For more general point sets, Akl et al.~\cite{akl2007planar} show connectivity for all small point sets with up to $8$~points in general position.
Aichholzer et al.~\cite{2023Aicholzer} confirm \Cref{conj:conn} for generalized double circles; a point set is a generalized double circle if it allows for a Hamiltonian cycle that intersects no segment between any two points.
While the obtained upper bound on the diameter is quadratic, they provide a linear upper bound of~$2n-4$ on the diameter for the special case of wheel sets of size $n$.

\subparagraph*{Plane spanning trees.}
Avis and Fukuda~\cite{AvisF96} show that the diameter of the flip graph of plane spanning trees is upper bounded by $2n-4$ for $n$ points in general position.
The best known lower bound of~$1.5 n-5$ due to Hernando et al.~\cite{HernandoHMMN99} even holds for the special case of convex point sets.
If one tree is a monotone path, Aichholzer et al.~\cite{aichholzer2022reconfiguration} show that it can be reconfigured into any other tree by $1.5n-2$ flips.

For points in convex position, the upper bound on the diameter has been improved several times.
Bousquet et al.~\cite{bousquet2023note} obtain a bound of~$2n-\Omega(\sqrt n)$ and conjecture $1.5n$ to be the tight bound.
Aichholzer et al.~\cite{aichholzer2022reconfiguration} bound the distance of two given trees $T_1$ and~$T_2$ on convex point sets by $2|T_1\setminus T_2|-\Omega(\log(|T_1\setminus T_2|))$.
Bousquet et al.~\cite{bousquet2024reconfigurationSoCG} break the barrier of 2 in the leading coefficient and show that there always exists a flip sequence between any pair of non-crossing spanning trees $T_1$ and~$T_2$ of length at most $c\,|T_1\setminus T_2|$ where $c\approx 1.95$.
As a lower bound in terms of the symmetric difference, they provide a pair of trees such that a minimal flip sequence has length $\nicefrac{5}{3}\,|T_1\setminus T_2|$.
Most recently, Bjerkevik et al.~\cite{trees} show a lower bound of~${\nicefrac{14}{9}\cdot n - \BigO(1)}$ and an upper bound of $\nicefrac{5}{3}\cdot n - \BigO(1)$.
\section{Fundamental concepts and tools}
\label{sec:defs}

Let $S$ be a set of~$n$ points in general position, i.e., no three points in $S$ are collinear.
A \emph{plane spanning path} on $S$ is a path {that visits each point of~$S$ exactly once}.
{Every edge of a path} is a straight-line segment between a pair of points, and no two edges intersect except at a common endpoint.
We write $uv$ or $vu$ for an edge between $u,v\in S$.
For any path $P$, the points of degree~$1$ are the \emph{endpoints} of~$P$.
We denote the set of all plane spanning paths of~$S$ by~$\paths(S)$.
A \emph{flip} on a path $P\in \paths(S)$ replaces an edge $e$ of~$P$ with another edge $f$ on $S$ to obtain a new path~$P'\in \paths(S)$; for an example, see~\Cref{fig:flip}.
The flip graph of~$\paths(S)$ contains a vertex for each path in~$\paths(S)$ and an edge between any two paths that can be transformed into one another by a flip.
Similarly, for each~${\mathcal P'\subset\paths(S)}$, the flip graph of~$\mathcal P'$ refers to the subgraph induced by $\mathcal P'$.

\subparagraph*{Paths with a fixed start.}
We focus in particular on plane spanning paths with a fixed endpoint.
For every $s\in S$, we denote the set of all paths with $s$ as an endpoint by $\paths(S,s)$.
\cref{fig:flipgraph-33,fig:flipgraph-c4,fig:flipgraph-c5} illustrate examples of induced flip graphs.
Typically, we consider an $st$-path of~$\paths(S,s)$ as oriented from $s$ to $t$, and call~$s$ the \emph{start} and~$t$ the \emph{end}.
{We illustrate these as black and white squares, respectively.}
It is easy to see that all edges that can be introduced by a flip are incident to $t$, see~\cref{fig:directed-edge-flip}.

\begin{figure}[htbp]
    \centering
    \includegraphics[page=1]{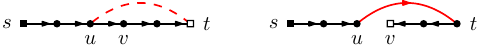}
    \caption{{A flip on a path with a fixed endpoint, illustrating~\cref{lem:fixed-flips}.}}
    \label{fig:directed-edge-flip}
\end{figure}

\begin{observation}
    \label{lem:fixed-flips}
    An $st$-path $P\in\paths(S,s)$ can be flipped to an $st'$-path ${P'\in\paths(S,s)}$ if and only if some edge $uv$ is replaced by $ut$,
    {i.e., an edge incident to $t$}.
\end{observation}

For two points $u,v\in S$ and a path $P$, we say  that~$u$ \emph{sees} $v$ if the straight-line segment~$uv$ does not intersect any of the edges of~$P$ except at a common endpoint.
Clearly, this relationship is symmetric, i.e., $u$ sees $v$ exactly if $v$ sees~$u$.
Then, \cref{lem:fixed-flips} implies that the degree of each path in the flip graph of~$\paths(S,s)$ is equal to the number of points that are visible from the endpoint of the path.

We say that a directed edge $uv$ is \emph{outgoing} for $u$ and \emph{incoming} for $v$.
For any directed path $P$ which contains the edge $uv$, we also say that~$u$ is the \emph{predecessor} of~$v$ on $P$, and~$v$ the \emph{successor} of~$u$ on $P$.
Moreover, let $\overleftarrow P$ denote the \emph{reversed} path of~$P$, i.e., $vu\in\overleftarrow P$ if and only if $uv\in P$.

\subparagraph*{Convex layers.}
We denote the convex hull of~$S$ by $\conv(S)$ and say that a point~${s\in S}$ is \emph{\out} if it lies on the boundary of the convex hull of~$S$.
A point that is not \out is \emph{inner}.
The \emph{layers} of~$S$ are defined as follows, see \cref{fig:convexlayer}:~${\layer_0\subseteq S}$ consists of all \out points.
For~$i\geq 1$, $\layer_i$ is the set of points that are \out in~$S\setminus \bigcup_{k<i} \layer_k$.
The number of non-empty layers in $S$ is the \emph{layer number}~$\layernumber(S)$.

\begin{figure}[htb]
    \begin{subfigure}{0.33\textwidth}%
        \centering
        \includegraphics[page=1]{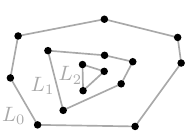}
        \subcaption{}
        \label{fig:convexlayer}
    \end{subfigure}%
    \begin{subfigure}{0.33\textwidth}%
        \centering
        \includegraphics[page=2]{figures/definitions}
        \subcaption{}
        \label{fig:edge-types}
    \end{subfigure}%
    \begin{subfigure}{0.33\textwidth}%
        \centering
        \includegraphics[page=3]{figures/definitions}
        \subcaption{}
        \label{fig:canonical-spiral}
    \end{subfigure}%
    \caption{A point set with (a) its layers, (b)  a chord, an inward cutting edge, and a level edge (from left to right), and (c) its counterclockwise spiral from $s\in \layer_0$.}
    \label{fig:definitions}
\end{figure}

We speak of an \emph{$\layer_i$-edge} if both endpoints are in~$\layer_i$.
For a convex point set, two points are \emph{adjacent} if they appear consecutively on the boundary of its convex hull.
We distinguish certain types of edges, as illustrated in~\cref{fig:edge-types}.
If $u$ and~$v$ are adjacent in~$\layer_i$, we say that~$uv$ is a \emph{level edge} of~$\layer_i$.
A \emph{chord} is then an $\layer_0$-edge that is not level.
Further, an edge $uv$ with $u\in \layer_i $ and~$ v\in \layer_j$ such that $i\neq j$ is called an \emph{inter-layer edge}; we call it  \emph{inward} if~${i<j}$, and \emph{outward} if $i>j$.
Finally, we say that an inter-layer edge is \emph{cutting}, if it crosses any level edge in~$S$.
For a path~$P\in\paths(S)$, we denote by $\leveledges_i(P)$ the number of level edges of~$\layer_i$ in $P$, and by $\cuttingedges(P)$ the number of cutting edges in $P$.

\subparagraph*{Suffix-independence.}
For a given $st$-path $P\in\paths(S,s)$, let $P[a,b]$ refer to the subpath of~$P$ from $a$ to $b$.
We say that~$P[a,b]$ is a \emph{prefix} of  $P$ if $a=s$, and a \emph{suffix} of~$P$ if $b=t$.
A~suffix~${P[a,t]}$ of $P$ is called \emph{\good}, if the convex hull of $P[a,t]$ does not contain any points of the prefix~$P[s,a]$.
By definition, $a$ is then \out with respect to the points in $P[a,t]$.
We say that a path $P\in\paths(S,s)$ is \emph{\suffixgood} if each suffix of~$P$ is \good, i.e., if every~$u\in S$ is \out with respect to the points in $P[u,t]$.
By definition, every suffix of a \suffixgood path is \suffixgood.
We denote by $\suffixgoodpaths(S,s)$ the set of all \suffixgood paths of~$\paths(S,s)$.

Note that the \suffixgood paths contain all paths that are monotone, radially monotone, or spirals.
We now present a sufficient criterion to identify a family of \suffixgood paths, the \emph{layer-monotone} paths.
A path on $S$ is layer-monotone if it visits the layers of~$S$ in increasing order, i.e., for all~${i\in[0,\layernumber(S)-1]}$, all points in layer $\layer_i$ are visited before any point in~$\layer_{i+1}$.

\begin{lemma}
    \label{lem:layer-monotone-paths}
    Let $S$ be a point set in general position, $s\in S$, and~${P\in\paths(S,s)}$.
    \begin{enumerate}[(1)]
        \item If $P$ is layer-monotone, then $P$  is \suffixgood.
        \item If $P$ is not \suffixgood, then it contains an outward edge.
    \end{enumerate}
\end{lemma}

\begin{proof}
    We show both statements.\\
    \descriptionlabel{(1)} Consider a suffix $P'$ of~$P$ that starts in $s'\in \layer_i$ for some $i$.
        Then, $P'$ visits all points from $\bigcup_{j>i} \layer_j$ and a subset $M$ of~$\layer_i$.
        Clearly, $s$ lies on the boundary of~$\bigcup_{j\geq i} \layer_j$ and thus in particular on the subset  $\bigcup_{j> i} \layer_j\cup M$.
        Moreover, no point in $\bigcup_{j\leq i} \layer_j\setminus M$ lies in the interior of~$\conv(P')$.
        Thus $P'$ is \good.\\
   \descriptionlabel{(2)} If $P$ is not \suffixgood, then by (a) it is not layer-monotone and thus has an outward edge.
\end{proof}

A special type of layer-monotone path, the~\emph{clockwise spiral} of~$S$ from an \out point $s$
is the path $s_1,\ldots, s_n$ such that~$s_1=s$ and~$s_{i+1}$ is an \out point of~${S\setminus\{s_1, \dots, s_{i-1}\}}$ which lies clockwise adjacent to $s_i$.
\Cref{fig:canonical-spiral} depicts a \emph{counterclockwise spiral}, {which can be constructed analogously}.
Finally, we define the \emph{$\layer_i$-suffix} of a path as the maximal suffix that contains only points on layers~$\layer_{j\geq i}$.

\subparagraph*{Useful properties.}
We present two useful properties for paths starting and ending in an \out point.
Two points $u,v\in S$ divide a layer ${\layer_{i}\setminus\{u,v\}}$ into two sets by the line through~$u$ and~$v$;
for a reference point {$w\in\layer_i$}, {we denote them by $\layer_i^+(u,v;w)$ and~$\layer_i^-(u,v;w)$} such that  $w\in \layer_i^+(u,v;w)$, see \cref{fig:layer-cuts}.

\begin{figure}[htb]
    \centering
    \includegraphics[page=1]{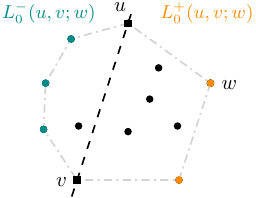}
    \caption{A pair of points divides each layer into $\layer_0^+(u,v;w)$ and~$\layer_0^-(u,v;w)$.}
    \label{fig:layer-cuts}
\end{figure}
By planarity of the paths, we obtain the following statement.
\begin{lemma}
    \label{lem:pseudochord}
    If an $st$-path $P$ starts and ends on $\layer_0$, then the prefix $P[s,v]$ for every ${v\in \layer_0 \setminus \{s,t\}}$ visits all points in~${\layer_0^-(s,v;t)}$.
\end{lemma}
\begin{proof}
    Suppose for a contradiction that~$P[v,t]$ visits a point $x\in \layer_0^-(s,v;t)$.
    Let~$\gamma$ denote the closed curve formed by $P[s,v]$, and {a Jordan curve from $v$ to~$s$ that intersects convex hull of~$P$ only in these two points}, see \cref{fig:usefulB,fig:usefulC}.
    Then~$\gamma$ separates $x$ from $t$ and thus, the Jordan curve theorem implies that~$P[x,t]$ intersects $P[s,v]$, a contradiction.
\end{proof}
\begin{figure}[htb]
    \begin{subfigure}[t]{0.48\textwidth}%
        \centering
        \includegraphics[page=2]{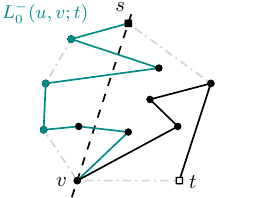}
        \subcaption{The prefix~$P[s,v]$ visits all points in $\layer_0^-(s,v;t)$.}
        \label{fig:usefulB}
    \end{subfigure}%
    \hfill%
    \begin{subfigure}[t]{0.48\textwidth}%
        \centering
        \includegraphics[page=3]{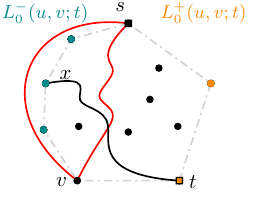}
        \subcaption{The closed curve $\gamma$ separates $x$ and $t$.}
        \label{fig:usefulC}
    \end{subfigure}%
    \caption{Illustration for \cref{lem:pseudochord}.}
    \label{fig:useful}
\end{figure}

This can be further extended to prove the following.
\begin{lemma}
    \label{obs:chordfree}
    If a path $P$ starts and ends in two adjacent points on $\layer_0$, then it is chord-free.
\end{lemma}
\begin{proof}
    Suppose an $st$-path $P$ contains a chord $uv$ and~$s,t\in \layer_0^+(u,v;t)$.
    According to \cref{lem:pseudochord}, $P[s,v]$ visits all points in $\layer_0^-(u,v;t)$.
    However, $P[s,u]$ does not visit any point in~${\layer_0^-(u,v;t)}$; otherwise it intersects $uv$, a contradiction.
\end{proof}
\section{Plane paths on point sets in convex position}
\label{sec:point-sets-in-convex-position}
In this section, we consider plane spanning paths on point sets in convex position, and present tight bounds for the radius and diameter of the flip graph. \cref{fig:flipgraph-c4,fig:flipgraph-c5} depict the flip graphs for $n=4$ and~$n=5$, respectively.

\begin{figure}[htb]
    \begin{subfigure}[t]{0.5\textwidth}
        \centering
        \includegraphics[page=1]{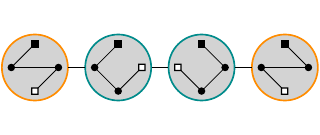}
        \subcaption{Flip graph of $\paths(S_1,s)$, for any ${s\in S_1}$.}
        \label{fig:flipgraph-c4}
    \end{subfigure}%
    \begin{subfigure}[t]{0.5\textwidth}
        \centering
        \includegraphics[page=2]{figures/ConvexFixedEnd}
        \subcaption{Flip graph of $\paths(S_2,s)$, for any ${s\in S_2}$.}
        \label{fig:flipgraph-c5}
    \end{subfigure}%
    \caption{We illustrate flip graphs on two point sets in convex position, $|S_1|=4$ and~$|S_2|=5$. The spirals are highlighted in cyan and the zigzag paths in orange.}
    \label{fig:flipgraph-convex}
    \medskip
\end{figure}

\convexBounds*

\begin{proof}
    Let $C_\text{ccw}$ and~$C_\text{cw}$ denote the {(counter-)clockwise} spirals of~$S$, i.e., the two plane spanning paths that only use convex hull edges, see~\cref{fig:convex2A,fig:convex2B}.

    \begin{figure}[htb]
    \begin{subfigure}[t]{0.166\textwidth}
        \centering
        \includegraphics[page=4]{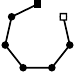}
        \subcaption{$C_\text{ccw}$}
        \label{fig:convex2A}
    \end{subfigure}%
    \begin{subfigure}[t]{0.166\textwidth}
        \centering
        \includegraphics[page=5]{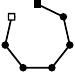}
        \subcaption{$C_\text{cw}$}
        \label{fig:convex2B}
    \end{subfigure}%
    \begin{subfigure}[t]{0.166\textwidth}
        \centering
        \includegraphics[page=6]{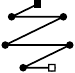}
        \subcaption{$Z_\text{ccw}$}
        \label{fig:convex2C}
    \end{subfigure}%
    \begin{subfigure}[t]{0.166\textwidth}
        \centering
        \includegraphics[page=7]{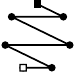}
        \subcaption{$Z_\text{cw}$}
        \label{fig:convex2D}
    \end{subfigure}%
    \begin{subfigure}[t]{0.166\textwidth}
        \centering
        \includegraphics[page=8]{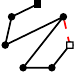}
        \subcaption{}
        \label{fig:convex2E}
    \end{subfigure}%
    \begin{subfigure}[t]{0.166\textwidth}
        \centering
        \includegraphics[page=9]{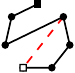}
        \subcaption{}
        \label{fig:convex2F}
    \end{subfigure}%
    \caption{Illustration for the proof of \cref{lem:convex-fixed-flipgraph}.}
    \label{fig:convex2}
\end{figure}

    For the upper bounds,
    we show that each path $P$ can be transformed into~$C_\text{cw}$ or $C_\text{ccw}$ with $n-3$ flips.
    Because $C_\text{cw}$ and~$C_\text{ccw}$ can be transformed into each other using one flip, we obtain an upper bound of~$2(n-3)+1=2n-5$ on the diameter and of~$n-2$ on the radius.

    It thus remains to show that each path $P$ can be transformed into $C_\text{cw}$ or~$C_\text{ccw}$ with $n-3$ flips.
    As each path contains at least two convex hull edges, we may assume without loss of generality that~$P$ and~$C_\text{cw}$ share at least two edges; otherwise we consider $C_\text{ccw}$.
    Note that the endpoint of~$P$ is incident to a convex hull edge $e$ which is not contained in $P$.
    We add $e$ to $P$ and remove the appropriate edge $f$, see \cref{fig:convex2E,fig:convex2F}.
    Observe that~$f$ lies on the convex hull if and only if  $s=t$.
    Hence, the number of convex hull edges increases until we end in  $C_\text{cw}$ after at most $n-3$ flips.

    For the lower bound on the diameter, we consider the zigzag-paths $Z_\text{cw}$ and~$Z_\text{ccw}$ depicted in \cref{fig:convex2C,fig:convex2D}.
    We  show they have distance $2n-5$; this directly implies that the radius is lower bounded by $\lceil(2n-5)/2\rceil=n-2$.
    Note that the edges incident to $s$ differ in $Z_\text{cw}$ and~$Z_\text{ccw}$.
    In order to flip this first edge,  by \cref{lem:fixed-flips}, the terminal must be a neighbor of~$s$; note that this is only possible if $P\in\{C_\text{cw},C_\text{ccw}\}$.
    Because $C_\text{cw}$ shares only two edges with $Z_\text{cw}$ (and only one with $Z_\text{ccw}$), it takes $n-3$ flips to transform $Z_\text{cw}$ into $C_\text{cw}$.
    Afterward, it takes at least $n-2$ flips to flip $C_\text{cw}$ to $Z_\text{ccw}$ because they share only one edge.

    By the above constructions, it is clear that the spirals are centers.
    As they lie on any shortest paths from $Z_\text{cw}$ to $Z_\text{ccw}$, they are the unique centers.
\end{proof}
\section{The subgraph of \suffixgood paths is connected}
\label{sec:the-subgraph-of-suffixgood-paths}
In this section, we investigate the subclass of \suffixgood paths.
We show that the flip graph on $\suffixgoodpaths(S,s)$ is connected, i.e., \cref{thm:suffix-good}.
In order to prove this theorem, the following property is crucial.
We say that a path is \emph{strongly \suffixgood} if both of its orientations are \suffixgood.

\begin{lemma}
    \label{lem:suffix-good-reversible-path}
    For any two \out points $s,t\in S$, there exists an $st$-path $P$ that is strongly \suffixgood.
\end{lemma}

In fact, \cref{thm:suffix-good,lem:suffix-good-reversible-path} also imply that the flip graph induced by the union of all \suffixgood paths is connected.

\begin{corollary}
    \label{cor:suffix-good} 
    For every point set $S$ in general position, the flip graph of the \suffixgood paths $\bigcup_{s\in S} \suffixgoodpaths(S,s)$ is connected.
\end{corollary}
\begin{proof}
    Let $P_1$ and~$P_2$ be two \suffixgood paths starting in $s_1$ and~$s_2$, respectively.
    By definition, $s_1$ and~$s_2$ are \out.
    If $s_1=s_2$, \cref{thm:suffix-good} implies that there exists a flip sequence from $P_1$ to $P_2$.
    So, we consider the case that~${s_1\neq s_2}$.
    Due to \cref{lem:suffix-good-reversible-path}, there exists a strongly \suffixgood path $P$ from $s_1$ to~$s_2$, i.e., $P$ and the reversed path $\overleftarrow P$ are \suffixgood.
    So \cref{thm:suffix-good} ensures that~$P_1$ can be flipped to $P$, and that $\overleftarrow P$ can be flipped to~$P_2$.
\end{proof}

We now present a constructive proof of \cref{lem:suffix-good-reversible-path}.
\bigskip
\begin{proof}[Proof of \cref{lem:suffix-good-reversible-path}]
    Without loss of generality, assume that~$s$ lies above~$t$ and that~$t$ is rightmost.
    Consider a vertical line $\verticalLine$ through $s$ that splits $S$ into two sets, the set $S_1$ of points contained in the left (closed) half-space of~$\verticalLine$ and its complement $S_2=S\setminus S_1$.
    We~define $P$ by two subpaths $P_1$ and~$P_2$, for an illustration see \cref{fig:unobstructed}.

    \begin{figure}[htb]
		\centering
		\includegraphics{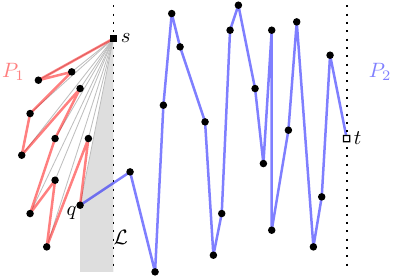}
		\caption{Illustration for the proof of \cref{lem:suffix-good-reversible-path}.}
		\label{fig:unobstructed}
    \end{figure}
    
    Starting from~$s$, $P_1$ visits all points in $S_1$ in the ccw order that~$s$ sees them.
    In particular, if $|S_1|>1$, we ensure that the first edge of~$P_1$ is a convex hull edge of~$S$; if $|S_1|=1$, $P_1$ is the trivial path starting and ending in~$s$.

    Let $q$ denote the endpoint of~$P_1$.
    Observe that the cone spanned by $\verticalLine$ and~$sq$ contains no points; otherwise such a point would be in $S_1$ and visited after $q$.
    The path~$P_2$ visits all points in $S_2\cup \{q\}$.
    It starts in $q$ and then collects the points by increasing $x$-coordinate.
    As~$t$ is rightmost, $P_2$ ends in $t$.
    Clearly, the concatenation of~$P_1$ and~$P_2$ is an $st$-path.
 	
    It remains to prove that~$P$ and~$\overleftarrow P$ are \suffixgood.
    By construction, for each~${p\in S_1}$, the line through $p$ and~$s$ separates $P[s,p]$ from $P[p,t]$.
    For each $p\in S_2$, the vertical line through $p$ separates $P[s,p]$ from $P[p,t]$.
    Hence, each $p\in S$ lies on the boundary of the convex hull of its reversed prefix and~suffix, and the specified lines separate their convex hulls.
    Consequently, for all $p\in S$, the reversed prefix and the suffix starting in $s$ is \good.
\end{proof}

Using~\cref{lem:suffix-good-reversible-path} as an essential tool, we now prove \cref{thm:suffix-good}.
\thmSuffix*
\begin{proof}
    We prove this by induction on the cardinality of~$S$.
    In the base case, for~${|S|=2}$, there exists exactly one path from $s$ to the other point of~$S$.
    Thus, the flip graph is trivially connected.

    For the induction step, consider a point set $S$ with $|S|>2$ and two arbitrary \suffixgood paths $P_1,P_2\in\suffixgoodpaths(S,s)$.
    Let $s_i$ be the successor of~$s$ on $P_i$, and let $P_i'$ be the suffix of~$P_i$ starting at~$s_i$.
    We distinguish two cases.
    
    If $s_1=s_2$, then by induction, $P'_1$ can be transformed into $P'_2$.
    This directly yields a flip sequence from $P_1$ to $P_2$.

    If $s_1\neq s_2$, we consider a path $H$ from $s_1$ to~$s_2$ on $S\setminus \{s\} $ such that each prefix and suffix is \good;
    the existence of $H$ is guaranteed by~\cref{lem:suffix-good-reversible-path}.
    By induction,~$P_1'$ can be transformed into $H$, see \cref{fig:suffix-independent-induction-a,fig:suffix-independent-induction-b}.
    We may replace the edge $ss_1$ with $ss_2$ in the resulting path, see \cref{fig:suffix-independent-induction-b,fig:suffix-independent-induction-c}.
    This reverses~$H$; however, as any prefix of~$H$ is \good, any suffix of~$\overleftarrow H$ is \good as well.
    Thus, by induction, $\overleftarrow H$ can be transformed into~$P_2'$, see \cref{fig:suffix-independent-induction-c,fig:suffix-independent-induction-d}.
    We conclude that the flip graph of the \suffixgood paths in $\suffixgoodpaths(S,s)$ is connected.
\end{proof}
\begin{figure}[htb]
    \hfil%
    \begin{subfigure}[b]{0.22\textwidth}
        \centering
        \includegraphics[page=1]{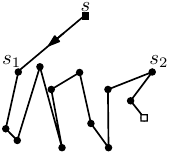}
        \subcaption{$P_1$}
        \label{fig:suffix-independent-induction-a}
    \end{subfigure}%
    \hfill%
    \begin{subfigure}[b]{0.22\textwidth}
        \centering
        \includegraphics[page=2]{figures/suffix-independent-induction}
        \subcaption{}
        \label{fig:suffix-independent-induction-b}
    \end{subfigure}%
    \hfill%
    \begin{subfigure}[b]{0.22\textwidth}
        \centering\includegraphics[page=3]{figures/suffix-independent-induction}
        \subcaption{}
        \label{fig:suffix-independent-induction-c}
    \end{subfigure}%
    \hfill%
     \begin{subfigure}[b]{0.22\textwidth}
        \centering
        \includegraphics[page=4]{figures/suffix-independent-induction}
        \subcaption{$P_2$}
        \label{fig:suffix-independent-induction-d}
    \end{subfigure}%
    \caption{
        Illustration for the proof of \cref{thm:suffix-good}.
        By induction, we can flip $P_1$ and~$P_2$ to a path that has $H$ or $\overleftarrow H$ as a suffix; indicated by the cyan subpath.
    }
    \label{fig:suffix-independent-induction}
\end{figure}
\section{General tools and point sets on two convex layers}
\label{sec:twolayers}

In this section, we consider point sets with two convex layers.
We start by investigating the setting with a fixed \out start point.

\thmConvOut*

The idea for our proof of \cref{thm:2conv} is based on the insight that if each path can be flipped to some \suffixgood path, then \cref{thm:suffix-good} implies the connectivity of the respective flip graph.
In a first step, we aim to obtain some path~${P\in \paths(S,s)}$ where the (outgoing) edge of~$s$ is a \emph{level} edge, i.e., it lies on the boundary of the convex hull.
If the first edge of~$P$ is a level edge $sq$, then the suffix $P[q,t]$ is \good, and we may consider this suffix instead (or assume by induction that we can flip to a \suffixgood path).

To obtain a path that starts with a level edge, we prove that we can increase~$\leveledges_0$ for any given path, unless $\leveledges_0=|\layer_0|-1$; recall that~$\leveledges_0(P)$ denotes the number of \mbox{$\layer_0$-level} edges in $P$.
It follows that we eventually introduce a level edge that is outgoing from $s$.
The crucial part lies in considering paths with two special properties.

In particular, a path $P\in \paths(S,s)$ with $s\in \layer_0$ belongs to $\chordfreepaths(S,s)$ if
\begin{enumerate}[(i)]
    \item \label{itm:i} it has no chords, and
    \item \label{itm:ii} the outgoing edge of~$s$ goes to some $q\notin \layer_0$.
\end{enumerate}

Furthermore, we say that the \emph{\prop} holds for $k\in\mathbb{N}$ exactly if for every point set~$S$ in general position with $L(S)\leq k$,
any $st$-path ${P\in\chordfreepaths(S,s)}$ with $\leveledges_0(P)<|\layer_0|-1$ can be flipped to an $st'$-path $P'\in\paths(S,s)$, such that 
(a)~$\leveledges_0(P')>\leveledges_0(P)$ if $t\in \layer_0$, or 
(b)~$\leveledges_0(P')=\leveledges_0(P)$ and $t'\in \layer_0$, otherwise.

\begin{theorem}
    \label{thm:k-prop-implies-connection}
    If the \prop holds, then for every   point set $S$ in general position  with $L(S)\leq k$ and~$s\in \layer_0$, the flip graph of~$\paths(S,s)$ is connected.
\end{theorem}
\begin{proof}
    By \cref{thm:suffix-good}, it suffices to show that each path can be flipped to some \suffixgood path.
    We prove this by induction on $|S|$.
    The induction base for $|S|=1$ is trivial because $\paths(S,s)$ consists of a single path.

    For the induction step, consider an $st$-path $P\in \paths(S,s)$ where $n\coloneqq|S|>1$ and assume that the statement holds for all point sets with less than $n$ points.

    First, we show that we can restrict our attention to chord-free paths.
    If the path~$P$ contains a chord, we consider the first chord $uv$ of~$P$.
    The \good suffix~$P[u,t]$ visits a subset $S'\subset S$ of fewer than $n$ points, so the induction hypothesis ensures that~$P[u,t]$ can be flipped to any path in $\paths(S',u)$.
    In particular, $P[u,t]$ (and by extension, $P$) can be flipped to a path $P'$ ending in $v$ (which exists by \cref{lem:suffix-good-reversible-path}).
    For an illustration, consider \cref{fig:good-prefixes}.
    \begin{figure}[htb]
        \begin{subfigure}[t]{0.33\textwidth}
        	\centering
        	\includegraphics[page=2]{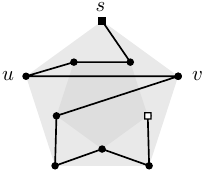}
        	\subcaption{The $st$-path $P$.}
        	\label{fig:good-prefixesA}
        \end{subfigure}%
        \begin{subfigure}[t]{0.33\textwidth}
        	\centering
        	\includegraphics[page=3]{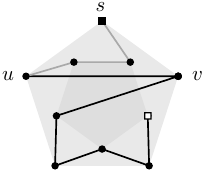}
        	\subcaption{The suffix $P[u,t]$.}
        	\label{fig:good-prefixesB}
        \end{subfigure}%
        \begin{subfigure}[t]{0.33\textwidth}
        	\centering
        	\includegraphics[page=5]{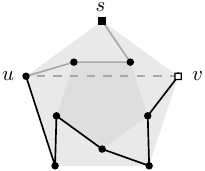}
        	\subcaption{The path $P'$.}
        	\label{fig:good-prefixesD}
        \end{subfigure}%
        \caption{Flipping to chord-free paths.}
        \label{fig:good-prefixes}
    \end{figure}
    The path~$P'$ is chord-free by \cref{obs:chordfree} and as it does not contain the edge $uv$ (which would be a chord in~$S$), adding the prefix $P[s,u]$ yields a chord-free path. 
    We can therefore assume property (\ref{itm:i}) for the remainder of this proof.

    We now further establish property~(\ref{itm:ii}).
    To this end, let $s_1$ denote a point such that~$ss_1$ is the first edge of~$P$.
    If $s_1\in\layer_0$ (implying that $ss_1$ lies on the boundary of the convex hull),
    we may flip the \good suffix~$P[s_1,t]$  to a \suffixgood path by the induction hypothesis.
    Consequently, the same applies to $P$.
    Thus, we may assume that $s_1\notin \layer_0$, i.e., property (\ref{itm:ii}) holds.

    It thus remains to show that each path $P$ in $\chordfreepaths(S,s)$ can be flipped to a \suffixgood path.
    By repeatedly using the \prop, we flip $P$ to a path where the first edge is a level edge in $\layer_0$.
    Then, as above, we may use induction to flip to a \suffixgood path.
    This concludes the proof.
\end{proof}

In order to prove \cref{thm:2conv}, it remains to establish the \propk{2}.
To this end, we introduce a crucial tool that exploits visibility between a suffix of a path and some preceding edge.
The following statement allows us to argue the existence of flip sequences to paths that end at arbitrary points which, in the original path, directly precede an \good~suffix, as illustrated in~\cref{fig:convex-region-idea}.

\begin{figure}[htb]
    \begin{subfigure}[b]{0.33\textwidth}
        \centering
        \includegraphics[page=1]{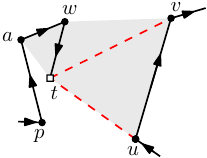}
        \subcaption{}
        \label{fig:convex-region-idea-a}
    \end{subfigure}%
    \begin{subfigure}[b]{0.33\textwidth}
        \centering
        \includegraphics[page=2]{figures/convex-region-idea}
        \subcaption{}
        \label{fig:convex-region-idea-b}
    \end{subfigure}%
    \begin{subfigure}[b]{0.33\textwidth}
        \centering
        \includegraphics[page=3]{figures/convex-region-idea}
        \subcaption{}
        \label{fig:convex-region-idea-c}
    \end{subfigure}%
    \caption{
        Illustration for \cref{lem:convex-region-chord-flip}:
        We exploit the visibility between an \good suffix $P[a,t]$ and some preceding edge $uv\in P$ (some flips omitted).
    }
    \label{fig:convex-region-idea}
\end{figure}

\begin{lemma}
    \label{lem:convex-region-chord-flip}
    Let $S$ be a point set in general position and~$P\in\paths(S,s)$ a path with an \good suffix $X=P[a,t]$ spanning $A\subset S$,
    such that the flip graph of~$\paths(A, a)$ is connected.
    If there exists an edge $uv\in P\setminus X$ such that
    (i) $uv$~is a boundary edge of $\conv(A\cup \{u, v\})$ and
    (ii) the interior of ${\conv(A\cup \{u, v\})\setminus(\conv(A)\cup \{uv\})}$ contains neither points nor parts of edges,
    then there exists a flip sequence from $P$ to a path that ends at the predecessor of~$a$ on $P$.
    This flip sequence introduces only edges inside $\conv(A\cup \{u, v\})$.
\end{lemma}

\begin{proof}
    Let $A_{uv}\coloneqq A\cup\{u,v\}$.
    Without loss of generality, we assume that~$uv$ is a counterclockwise edge on the boundary of~$\conv(A_{uv})$, and that~$v$ is not the predecessor of $a$ on $P$.
    If~${va\in P}$, clearly $P[u,t]$ is an independent suffix, which implies the existence of an applicable flip sequence into a path that ends at $v$.

    Let~$c_1,\dots, c_k$ refer to the counterclockwise chain of \out points of~$\conv(A)$ such that~$vc_1$ and~$c_{k}u$ are counterclockwise edges on the boundary of~$\conv(A_{uv})$.
    Note that $c_1$ and $c_k$ see $v$ and $u$, respectively.
    Let~${z_1,\dots, z_\ccwchain}$ then be the remaining \out points of $A$ (which are not \out points of $A_{uv}$) such that~${z_1=c_k}$ and~$z_\ccwchain=c_1$.
    For an illustration, consider \Cref{fig:convex-region-exploit-overview}.
    
    \begin{figure}[htb]
        \centering
        \includegraphics[page=1]{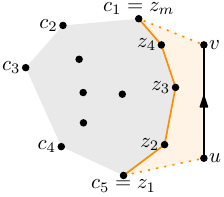}
        \caption{Illustration of the set $A$ and the chains $c_1,\dots, c_k$ and~$z_1,\dots, z_\ccwchain$.}
        \label{fig:convex-region-exploit-overview}
    \end{figure}
    Note that since the orange region is not intersected by any edge and the path is planar, $a$ is not in the chain~${z_2,\dots, z_{\ccwchain-1}}$.

    We now construct the flip sequence by induction on $|A|$.
    For the base case of~$|A|=1$, the suffix $X$ consists of the point $a$.
    By assumption, the triangle~$\Delta(a,u,v)$ does not contain points, and $t$ sees both~$u$ and~$v$, as depicted in~\cref{fig:convex-region-base-case-a}.
    We~perform two flips; we replace $uv$ with~$ua$, see~\cref{fig:convex-region-base-case-b}.
    This reorients the edge between $a$ and its predecessor~$p$ on $P$.
    Further, we remove $ap$ and add $av$, obtaining a path that ends in $p$, as illustrated in~\cref{fig:convex-region-base-case-c}.

    \begin{figure}[htb]
        \begin{subfigure}[b]{0.33\textwidth}
            \centering
            \includegraphics[page=1]{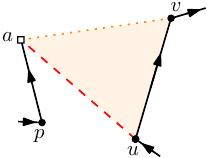}
            \subcaption{}
            \label{fig:convex-region-base-case-a}
        \end{subfigure}%
        \begin{subfigure}[b]{0.33\textwidth}
            \centering
            \includegraphics[page=2]{figures/convex-region-base-case}
            \subcaption{}
            \label{fig:convex-region-base-case-b}
        \end{subfigure}%
        \begin{subfigure}[b]{0.33\textwidth}
            \centering
            \includegraphics[page=3]{figures/convex-region-base-case}
            \subcaption{}
            \label{fig:convex-region-base-case-c}
        \end{subfigure}%
        \caption{
            Illustration for the proof of \cref{lem:convex-region-chord-flip}:
            the base case, where $|A|=1$.
        }
        \label{fig:convex-region-base-case}
    \end{figure}

    For the induction step, we assume that the claim holds true for \good suffixes on at most $|A|-1$ points.
    It therefore suffices to  find a flip sequence that removes one point from the suffix and yields another \good suffix together with a corresponding visible edge.
    We distinguish two cases based on the length of the chain $z_1,\dots, z_\ccwchain$.
    Note that~$|A|\geq 2$ implies~${\ccwchain\geq 2}$.

    If $\ccwchain=2$, at least one of~$z_1,z_2$ is not the start point $a$ of~$X$; we denote it by $z$.
    Clearly, $z$ sees both $u$ and~$v$.
    By \cref{lem:suffix-good-reversible-path}, there exists a strongly \suffixgood $az$-path on~$A$.
    As  the flip graph of~$\paths(A,a)$ is connected by assumption, we may assume without loss of generality that $P$ has a strongly \suffixgood suffix $P[a,z]$.
    Let $p$ refer to the predecessor of~$z$ in $P$, as illustrated in \cref{fig:convex-region-ell-2-a}.
    As above, we may replace $uv$ with~$uz$, reorienting $pz$ and allowing us to replace it with $zv$.
    The result is an $sp$-path $P'$ with a strongly \suffixgood suffix $P'[a,p]$ of size $|A|-1$, see \cref{fig:convex-region-ell-2-b}.
    As $z_1,u,v,z_2$ form a counterclockwise chain of points on the boundary of~$\conv(A_{uv})$, the two diagonals $z_{1}v$ and~$uz_2$ of the quadrilateral lie on the boundary of~$\conv(A\cup\{v\})$ and~$\conv(A\cup\{u\})$, respectively.
    Thus, we may apply the induction hypothesis to flip to a path that ends at the predecessor of $a$ on $P$.
    For an illustration of the entire flip sequence, see~\Cref{fig:convex-region-ell-2}, particularly the application of the induction hypothesis in~\Cref{fig:convex-region-ell-2-c}.

    \begin{figure}[htb]
        \begin{subfigure}[b]{0.33\textwidth}
            \centering
            \includegraphics[page=5]{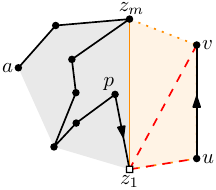}
            \subcaption{}
            \label{fig:convex-region-ell-2-a}
        \end{subfigure}%
        \begin{subfigure}[b]{0.33\textwidth}
            \centering
            \includegraphics[page=6]{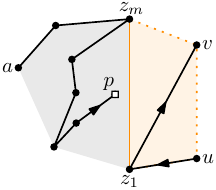}
            \subcaption{}
            \label{fig:convex-region-ell-2-b}
        \end{subfigure}%
        \begin{subfigure}[b]{0.33\textwidth}
            \centering
            \includegraphics[page=7]{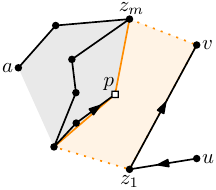}
            \subcaption{}
            \label{fig:convex-region-ell-2-c}
        \end{subfigure}%
        \caption{
            Illustration for the proof of \cref{lem:convex-region-chord-flip} for the case $\ccwchain=2$.
        }
        \label{fig:convex-region-ell-2}
    \end{figure}

    If $\ccwchain>2$, we need to work a bit harder.
    Firstly, we observe that there  exists~${i\notin\{ 1, \ccwchain\}}$ such that the point $z_i$ can see the edge $uv$.
    To see this, note that by convexity, it holds that if $z_i$ sees $v$ then $z_{i+1},\dots,z_\ccwchain$ see $v$, and similarly if $z_i$ sees $u$ then $z_1,\dots,z_{i-1}$ see $u$.

    Secondly, we show that we may assume that the predecessor of~$z_i$ is either~$z_{i-1}$ or $z_{i+1}$.
    Clearly, at least one of~$z_{i-1}$ and~$z_{i+1}$ is different from $a$; we denote it by~$p$.
    By \cref{lem:suffix-good-reversible-path}, there exists a strongly \suffixgood $ap$-path on $A\setminus\{z_i\}$.
    Together with the boundary edge $pz_i$, we obtain an $az_i$-path on $A$.
    Because $\paths(A, a)$ is connected by assumption, we may flip $X$ to this path.
    Consequently, we may assume that the predecessor $p$ of~$z_i$ is either $z_{i-1}$ or $z_{i+1}$, which see $u$ or $v$, respectively, as illustrated in \cref{fig:convex-region-ell-g2-a}.

    We then apply the following flips.
    First we replace the edge $uv$ with $uz_i$, and then $pz_i$ with $z_iv$.
    This yields a path that ends in $p$, see \cref{fig:convex-region-ell-g2-b}.

    \begin{figure}[htb]
        \begin{subfigure}[b]{0.33\textwidth}
            \centering
            \includegraphics[page=2]{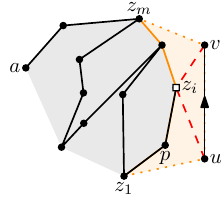}
            \subcaption{}
            \label{fig:convex-region-ell-g2-a}
        \end{subfigure}%
        \begin{subfigure}[b]{0.33\textwidth}
            \centering
            \includegraphics[page=3]{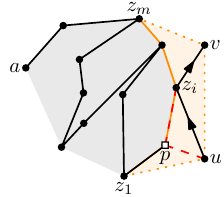}
            \subcaption{}
            \label{fig:convex-region-ell-g2-b}
        \end{subfigure}%
        \begin{subfigure}[b]{0.33\textwidth}
            \centering
            \includegraphics[page=4]{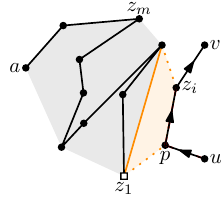}
            \subcaption{}
            \label{fig:convex-region-ell-g2-c}
        \end{subfigure}%
        \caption{
            Illustration for the proof of \cref{lem:convex-region-chord-flip} for the
            case $\ccwchain>2$.
        }
        \label{fig:convex-region-ell-g2}
    \end{figure}

    As $p$ sees either $u$ or $v$, we may replace $z_iv$ with $z_ip$ and one of~$pu$ or $pv$, respectively.
    This yields a path ending in the predecessor of~$p$, as shown in \cref{fig:convex-region-ell-g2-c}.

    Finally, because $z_ip$ lies in the boundary of~$\conv(A)$, we may apply the induction hypothesis to the current suffix starting in $a$ and the edge $z_ip$.
    This completes the proof.
\end{proof}

We use \cref{lem:convex-region-chord-flip}  to prove the \prop for $k=2$.

\begin{lemma}
    \label{thm:l0-level-increase}
    The \prop holds for $k\leq 2$.
\end{lemma}

\begin{proof}
    For $k=1$, this follows from~\cref{lem:convex-fixed-flipgraph}.
    For $k=2$, the idea is as follows.
    In a first phase, we flip $P$ to a path~$P'$ that ends in~$\layer_1$ with $\leveledges_0(P') > \leveledges_0(P)$.
    Then, in a second phase, we flip to a path~$P''$ that ends on~$\layer_0$ such that~$\leveledges_0(P'') \geq \leveledges_0(P')$.
    For an example, where each phase consists of a single flip, consider~\cref{fig:l0-to-l1}.

    \begin{figure}[htb]
        \centering
        \begin{subfigure}[b]{0.33\textwidth}
            \centering
            \includegraphics[page=1]{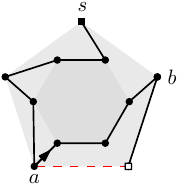}
            \subcaption{}
            \label{fig:l0-to-l1A}
        \end{subfigure}%
        \begin{subfigure}[b]{0.33\textwidth}
            \centering
            \includegraphics[page=2]{figures/layer-cuts-and-pseudochords}
            \subcaption{}
            \label{fig:l0-to-l1B}
        \end{subfigure}%
        \begin{subfigure}[b]{0.33\textwidth}
            \centering
            \includegraphics[page=3]{figures/layer-cuts-and-pseudochords}
            \subcaption{}
            \label{fig:l0-to-l1C}
        \end{subfigure}%
        \caption{Illustration for the proofs of \cref{lem:level-edge-l0-to-l1,lem:level-edge-l1-to-l0-base-case}.
        }
        \label{fig:l0-to-l1}
    \end{figure}

    We start with the first phase, consisting of a single, straightforward flip.

    \begin{claim}
        \label{lem:level-edge-l0-to-l1}
        An $st$-path $P\in\chordfreepaths(S,s)$ with $s,t\in \layer_0$
        can be flipped to a chord-free $st'$-path $P'\in\paths(S,s)$ with $t'\in \layer_1$ such that~$\leveledges_0(P')>\leveledges_0(P)$.
    \end{claim}

    \begin{claimproof}
        It suffices to show that one of the two points adjacent to $t$ in $\layer_0$ has an outgoing edge to~$\layer_1$.
        We denote these two adjacent points by $a$ and $b$, respectively, see~\cref{fig:l0-to-l1A}.
        Without loss of generality, we may assume that~$P$ visits $a$ before~$b$.
        If $a=s$, we are done: Recall that due to $P\in\chordfreepaths(S,s)$, the outgoing edge of $s$ goes to a point in $\layer_1$.
        Assume now that $a\neq s$.
        Then, by \cref{lem:pseudochord}, $P[s,a]$ visits all points in $\layer_0^-(s,a;t)$.
        In particular, the outgoing edge of~$a$ does neither end at its adjacent points in $\layer_0$ ($t$ is excluded because $P$ visits $b$ between $a$ and~$t$ and the other point is excluded because~$P[s,a]$ visits all points in~${\layer_0^-(s,a;t)}$), nor elsewhere in $\layer_0$ due to chord-freeness of~$P$.
        Thus, the outgoing edge of~$a$ ends in some $p\in \layer_1$ as depicted in~\cref{fig:l0-to-l1A}.
        Therefore, we flip the non-level edge $ap$ to the level edge $at$, yielding a path~$P'$, see~\cref{fig:l0-to-l1B}.
        As a result, $P'$ ends on $p\in \layer_1$ and has one more level edge in $\layer_0$.
    \end{claimproof}

    Now, we turn our attention to the second phase, which is a little more intricate.
    If $\leveledges_0(P')$ is maximal, then $P'$ is \suffixgood and we are done.
    Otherwise, $\leveledges_0(P')<|\layer_0|-1$.
    Thus, by \cref{lem:layer-monotone-paths}, there exists an outward edge~$pq$.

    For simplicity, we denote the current path by $P$ and its end by $t$ throughout the second phase.
    We distinguish three cases:
    (1) $t$ sees $p$,
    (2) the visibility is blocked by an $\layer_1$-edge, or
    (3) the visibility is blocked by a cutting edge.
    In cases~(1) and (2), we are able to flip to a path ending in $\layer_0$ as desired.
    In case~(3), we reduce the number of cutting edges and will eventually end in case~(1) or (2).
    It thus remains to consider these three cases.

    \begin{description}
        \item[Case (1):] $t$ sees $p$.
    \end{description}

    \begin{claim}
        \label{lem:level-edge-l1-to-l0-base-case}
        Let $P\in\chordfreepaths(S,s)$ be a path from $s\in \layer_0$ to $t\in \layer_1$.
        If $t$ sees a point $p\in \layer_1$ with an outgoing edge to $q\in \layer_0$, then~$P$ can be flipped to an $st'$-path $P'\in\paths^*(S,s)$ such that~$t'\in \layer_0$ and~$\leveledges_0(P')=\leveledges_0(P)$.
    \end{claim}
    \begin{claimproof}
        We replace the edge $pq$ by $pt$, resulting in a path that ends in $q\in \layer_0$.
        As~the edge $pq$ is an inter-layer edge, we have $\leveledges_0(P')=\leveledges_0(P)$.
        Moreover, since $pt$ is an $\layer_1$-edge, $P'$ is chord-free.
        For illustrations, see \cref{fig:l0-to-l1B,fig:l0-to-l1C}.
    \end{claimproof}

    Otherwise, if $t$ does not see a point $p$ with an outward edge $pq$, the segment~$tp$ intersects edges of~$P$ which obstruct the visibility.
    Let $e=uv$ be the first obstructing edge along $tp$.

    Recall that in this case, the {$\layer_1$-suffix} of a path is the maximal suffix that contains only points on~$\layer_1$.
    Without loss of generality, we may assume that~$e$ does not belong to the $\layer_1$-suffix $X$ of~$P$.
    Otherwise, there exist two points $x_1,x_2$ on $X$ such that~$\layer_1^-(x_1,x_2; p)\cup\{x_1,x_2\}$ contains
    the entire $\layer_1$-suffix $X$.
    Thus, we may use \Cref{lem:convex-fixed-flipgraph} on $X$ to flip to a path which ends at $x_1$ or~$x_2$.
    Consequently, no edge of~$X$ may obstruct the visibility to $p$.
    For later reference, it is important to note that this step only introduces $\layer_1$-edges because $X$ lives on a subset of~$\layer_1$.

    Moreover, note that at most one of~$u$ and~$v$ belongs to $\layer_0$ due to the chord-freeness of~$P$; further, if $u\in \layer_1, v\notin \layer_0$, then clearly $t$ sees $u$ and we obtain a scenario as in case~(1).
    Thus, it remains to consider the cases (2) $u,v\in \layer_1$, and~(3)~${u\in \layer_0, v\in \layer_1}$.

    \begin{description}
        \item[Case (2):] The first obstructing edge $uv$ has $u,v\in \layer_1$ (and~$u,v\notin X$).
    \end{description}

    We observe that~$t$ sees both $u$ and~$v$ because $\layer_1$ is convex, and~$uv$ is the first obstructing edge.
    Let $A$ denote the points of~$X$, and~$a\in \layer_1$ the start of~$X$.
    By \cref{lem:convex-fixed-flipgraph} and convexity of~$\layer_1$, the flip graph of~$\paths(A,a)$ is connected.
    We can therefore apply \cref{lem:convex-region-chord-flip} to $X$ and~$uv$ to obtain a
    path that ends on $\layer_0$.
    This only modifies edges within $\layer_1$, and hence does not affect the number of level edges in $\layer_0$.

    \begin{description}
        \item[Case (3):] The first obstructing edge $uv$ has $u\in \layer_0$, $v\in \layer_1$ (and~$u,v\notin X$).
    \end{description}

    We can assume, without loss of generality, that the $\layer_1$-suffix $X$ of~$P$ contains all points of~$\layer_1^+(u,v;t)$.
    This stems from the fact that, if there was a point in~${\layer_1^+(u,v;t)}$ that is not part of~$X$, there exists a (potentially different) point~$p'$ in $\layer_1^+(u,v;t)$ with an outward edge.
    As a result, we could then select~$p'$ as a candidate for $p$ that is not obstructed by $uv$, which leads  us to either case~(1)~or~(2).

    We show how to end in a path with fewer cutting edges.
    This will bring us to one of the cases (1)--(3).
    As the number of cutting edges is bounded, we will eventually end in case (1) or (2), both of which imply the claim.
    Let $\cuttingedges(P)$ denote the number of cutting edges of~$P$.

    \begin{claim}
        \label{lem:cutting-edge-removal}
        Let $P\in\chordfreepaths(S,s)$ be an $st$-path with $s\in \layer_0$, $t\in \layer_1$ and~$uv$ an inward cutting edge.
        If the $\layer_1$-suffix $X$ of~$P$ contains all points of~$\layer_1^+(u,v;t)$, then $P$ can be flipped to a path $P'\in\chordfreepaths(S,s)$ with $\leveledges_0(P')\geq \leveledges_0(P)$ and~$\cuttingedges(P')<\cuttingedges(P)$.
    \end{claim}
    \begin{claimproof}
        As $uv$ is an inward cutting edge, we have  $u\in \layer_0$ and~$v\in \layer_1$.
        First we show that~${s\in L^-_0(u,v;t)}$, or $u=s$:
        Suppose $s\in L^+_0(u,v;t)$.
        Because $P\in\chordfreepaths(S,s)$, the path~$P$ starts with an inward edge.
        Hence, there exists a point in $L^+_1(u,v;t)$ which is not part of~$X$, a contradiction.

        \begin{figure}[htb]
            \centering
            \begin{subfigure}[b]{0.33\textwidth}
                \centering
                \includegraphics[page=5]{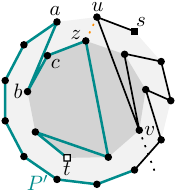}
                \subcaption{}
                \label{fig:cutting-edge-e}
            \end{subfigure}%
            \hfil%
            \begin{subfigure}[b]{0.33\textwidth}
                \centering
                \includegraphics[page=6]{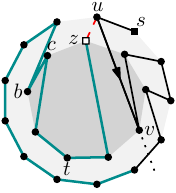}
                \subcaption{}
                \label{fig:cutting-edge-f}
            \end{subfigure}%
            \hfil%
            \begin{subfigure}[b]{0.33\textwidth}
                \centering
                \includegraphics[page=7]{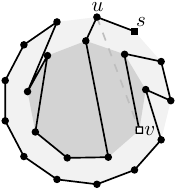}
                \subcaption{}
                \label{fig:cutting-edge-g}
            \end{subfigure}%
            \hfil%
            \\
            \hfil%
            \begin{subfigure}[b]{0.33\textwidth}
                \centering
                \includegraphics[page=2]{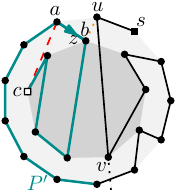}
                \subcaption{}
                \label{fig:cutting-edge-b}
            \end{subfigure}%
            \hfil%
            \begin{subfigure}[b]{0.33\textwidth}
                \centering
                \includegraphics[page=3]{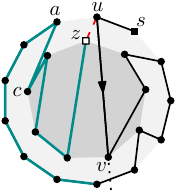}
                \subcaption{}
                \label{fig:cutting-edge-c}
            \end{subfigure}%
            \hfil%
            \begin{subfigure}[b]{0.33\textwidth}
                \centering
                \includegraphics[page=4]{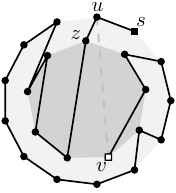}
                \subcaption{}
                \label{fig:cutting-edge-d}
            \end{subfigure}
            \hfil%
            \caption{Illustration for \cref{lem:cutting-edge-removal}.}
            \label{fig:cutting-edge}
        \end{figure}

        We define $P'$ as the longest suffix of~$P$ which does not cross the extension of~$uv$, see \Cref{fig:cutting-edge-b,fig:cutting-edge-e}.
        Note that~$P'$ contains exactly one inward edge $ab$; otherwise there exists a point in $L^+_1(u,v;t)$ which is not part of~$X$.
        We remark that~$ab$ may coincide with $uv$.
        Thus, $P'$ visits all points in $\layer_0^+(u,v;t)$ before going to $\layer_1$.
        Consequently, $P'$ is \suffixgood, see~\Cref{fig:cutting-edge-b,fig:cutting-edge-e}.
        Let $z$ be the point in $\layer_1^+(u,v;t)$ such that~$uz$ is not a cutting edge and~$u$ has a neighbor in $\layer_1$ which belongs to $\layer_1^-(u,v;t)$.

        We aim to flip $P'[a,t]$ to a path that ends in $z$.
        If $\layer_1^+(u,v;t)=\{z=b\}$, we are done.
        If $b\neq z$, as in \Cref{fig:cutting-edge-e}, then we flip $P'[b,t]$ (and thus $P$) to a path that ends at $z$ by \Cref{lem:convex-fixed-flipgraph,lem:suffix-good-reversible-path}.
        The resulting path is depicted in \Cref{fig:cutting-edge-f}.
        If $b=z$, as in \Cref{fig:cutting-edge-b}, we consider $P'[a,t]$.
        In this case, $a$ sees a point $c\in \layer_1^+(u,v;t)$ such that~$ac$ is not a cutting edge.
        Using \Cref{lem:convex-fixed-flipgraph,lem:suffix-good-reversible-path}, we flip the suffix $P'[b,t]=X$ to a path that ends in $c$.
        As $a$ then sees $c$, we may replace the edge $ab$ with $ac$, which is not a cutting edge.
        We~obtain a path with endpoint $z=b$, see \Cref{fig:cutting-edge-c}.

        Finally, we replace $uv$ by $uz$, see \Cref{fig:cutting-edge-d,fig:cutting-edge-g}.
        This does not affect the total number of~$\layer_0$-level edges, but reduces the number of cutting edges.
    \end{claimproof}
    This finishes the proof of~\Cref{thm:l0-level-increase}.
\end{proof}

Together, \cref{thm:k-prop-implies-connection,thm:l0-level-increase} imply  \cref{thm:2conv}.
\thmConvOut*

With similar arguments, we can deduce~\cref{cor:2conv}, confirming \cref{conj:conn} for point sets with layer number up to two.
We proceed as follows.

\thmFlipGraphConnected*

\begin{proof}
    By \cref{thm:2conv}, the flip graph of~$\paths(S,p)$ is connected for every ${p\in \layer_0}$.
    As there~exists a $pq$-path in $\paths(S,p)$ for any $q\in \layer_0$ by \cref{lem:suffix-good-reversible-path}, the flip graph of~$\bigcup_{p\in \layer_0}\paths(S,p)$ is connected.

    It remains to argue that any {$st$-path~$P{\in\paths(S)}$} with $s,t\in \layer_1$ can be flipped to a path with at least one end in $\layer_0$.
    To this end, we view $P$ as a directed path from $s$ to $t$ for the remainder of this proof.
    Let~$X$ refer to the $\layer_1$-suffix of~$P$.
    We distinguish three cases:
    \begin{description}
        \item[(1)] If $t$ sees $s$, we add the edge $st$ and delete any outgoing edge of an \out~point.
        \item[(2)] If $t$ sees a point $p$ such that~$pq\in P$ with $q\in \layer_0$, we replace $pq$ with $tp$.
        \item[(3)] If $t$ sees neither $s$ nor $p$, there exists an obstructing edge.
        We argue as~follows:
    \end{description}
    \vspace{-0.7em}

    Let $e=uv$ be the first obstructing edge along $ts$.
    Without loss of generality, we assume that~$e$ does not belong to~$X$.
    Otherwise, there exist two points $x_1,x_2$ on $X$ such that the entire $\layer_1$-suffix~$X$ lives on $\layer_1^-(x_1,x_2; t)\cup\{x_1,x_2\}$.
    We may thus use \Cref{lem:convex-fixed-flipgraph} on $X$ to flip to a path which ends at $x_1$ or~$x_2$.
    Consequently, no edge of~$X$ may obstruct the visibility between $s$ and $t$.

    If~$t$ can see both $u$ and~$v$, we apply \cref{lem:convex-region-chord-flip} to the \good suffix $X$ of $P$ and $uv$, flipping $P$ to a path that ends on~$\layer_0$.

    If $t$ cannot see one of~$u$ and~$v$, then $uv$ is either a chord or a cutting edge.
    If $uv$ is a chord, then we use \cref{thm:2conv} to flip the \good path $P[u,t]$ to end on an \out point.
    It remains to argue that the same is possible if $uv$ is a cutting edge.
    We argue analogously to the proof of~\cref{thm:l0-level-increase} in \cref{lem:cutting-edge-removal}.

    Assume that $uv$ is an inward cutting edge, i.e., $u\in \layer_0$ and~$v\in \layer_1$.
    There either exists a point $p$ with an outgoing edge to some $q\in\layer_0$ that is not obstructed by $uv$, or we can find a flip sequence to remove the edge $uv$ from the~path.
    In the first case, we flip the outgoing edge of $p$ and are done.

    We therefore proceed by showing how to remove $uv$ from $P$ and refer to the illustrations in~\cref{fig:cutting-edge} for an intuition, as we argue analogously.
    We can assume, without loss of generality, that the $\layer_1$-suffix $X$ of~$P$ contains all points of~$\layer_1^+(u,v;t)$.
    This stems from the fact that, if there was a point in~${\layer_1^+(u,v;t)}$ that is not part of~$X$, there exists a (potentially different) point~$p'$ in $\layer_1^+(u,v;t)$ with an outward edge.
    As a result, we could then select~$p'$ as a candidate for $p$ that is not obstructed by $uv$.
    It follows that $s\in L^-_1(u,v;t)$.

    We define $P'$ as the longest suffix of~$P$ which does not cross the extension of~$uv$.
    Note that~$P'$ contains exactly one inward edge $ab$; otherwise there exists a point in $L^+_1(u,v;t)$ which is not part of~$X$.
    We remark that~$ab$ may coincide with~$uv$.
    Thus, $P'$ visits all points in $\layer_0^+(u,v;t)$ before going to $\layer_1$.
    Consequently,~$P'$ is \suffixgood.
    Let $z$ be the point in $\layer_1^+(u,v;t)$ such that~$uz$ is not a cutting edge and~$u$ has a neighbor in $\layer_1^-(u,v;t)$.

    We aim to flip $P'[a,t]$ to a path that ends in $z$.
    If $\layer_1^+(u,v;t)=\{z=b\}$, we are done.
    If~${b\neq z}$, then we flip $P'[b,t]$ (and thus $P$) to a path that ends at~$z$ by \Cref{lem:convex-fixed-flipgraph,lem:suffix-good-reversible-path}.
    If $b=z$, we consider $P'[a,t]$.
    In this case, $a$ sees a point $c\in \layer_1^+(u,v;t)$ such that~$ac$ is not a cutting edge.
    Using \Cref{lem:convex-fixed-flipgraph,lem:suffix-good-reversible-path}, we flip the suffix $P'[b,t]=X$ to a path that ends in $c$.
    As $a$ then sees $c$, we may replace the edge $ab$ with $ac$, which is not a cutting edge.
    We obtain a path with endpoint $z=b$ and finally replace $uv$ by $uz$.
    This does not affect the total number of~$\layer_0$-level edges, but reduces the number of cutting edges by removing $uv$.

    We have shown that any {$st$-path~$P{\in\paths(S)}$} with $s,t\in \layer_1$ can be flipped to a path with at least one end in $\layer_0$, and
    thus conclude that the flip graph of~$\paths(S)$ is connected for point sets in general position with  $\layernumber(S)\leq 2$.
\end{proof}
\section{Discussion}
\label{sec:discussion}

In this work, we present progress towards proving \cref{conj:conn,conj:connFixed}.
In particular, we provide first results concerning the connectivity of flip graphs induced by plane spanning paths with a fixed end.
Moreover, we show that the subgraph induced by \suffixgood paths is connected.
Lastly, we present tools to allow for the extension of known connectivity results, based on empty convex regions within spanning paths, and show how to conclude connectivity for point sets with up to two convex layers.
Our insights indicate that investigating paths with a fixed end is a fruitful strategy.

We are optimistic that our tools prove to be useful to settle \cref{conj:conn,conj:connFixed}.
Specifically, we believe that \cref{lem:convex-region-chord-flip} could enable an inductive approach to settling the conjectures, based on the layer number of point sets.
Our application of \cref{lem:convex-region-chord-flip} to extend results from convex position (\cref{lem:convex-fixed-flipgraph}) to two convex layers (\cref{thm:2conv,cor:2conv}) could be viewed as a first step.

\bibliography{bibliography}

\end{document}